\documentclass[10pt,a4paper,english]{llncs}
\usepackage[T1]{fontenc}
\usepackage[utf8]{luainputenc}
\usepackage{amsmath}
\usepackage{amssymb}

\makeatletter

\pdfpageheight\paperheight
\pdfpagewidth\paperwidth

\numberwithin{equation}{section}
\numberwithin{figure}{section}

\date{}

\usepackage{tikz}
\usetikzlibrary{shapes,arrows,patterns,mindmap,trees,positioning,fit,calc}

\usepackage{subfigure}

\usepackage{bm}

\usepackage{lipsum}

\makeatother

\usepackage{babel}
\begin{document}

\title{Dining Cryptographers are Practical}

\author{Christian Franck\qquad{}Jeroen van de Graaf}

\institute{~}
\maketitle
\begin{abstract}
The dining cryptographers protocol provides information-theoretically
secure sender and recipient untraceability. However, the protocol
is impractical because a malicious participant who disrupts the communication
is hard to detect. We propose a scheme with optimal collision resolution,
in which computationally limited disruptors are easy to detect.
\end{abstract}

\section{Introduction}

The Dining Cryptographers protocol \cite{chaum1988dcp} is a special
primitive for anonymous communication in which senders and recipients
are unconditionally untraceable. Unlike relay-based techniques like
mixing or onion routing, it requires no assumption on the network,
no cryptographic assumptions and no third party. Because of these
advantages it could be useful in many scenarios like electronic voting,
low latency anonymous communication or multiparty computation.

During a typical round of the protocol, the participants $P_{1},...,P_{n}$
respectively broadcast the ciphertexts $O_{1},...,O_{n}$. Each ciphertext
$O_{i}$ looks like a random value, but the sum of all the ciphertexts
$C=\sum_{i=1}^{n}O_{i}$ reveals an anonymous message $M$ (i.e.,
$C=M$). The sender remains unknown; that is, each participant could
be the sender of the message. The protocol typically comprises two
steps:
\begin{enumerate}
\item During the first step, each pair of participants $P_{i}$ and $P_{j}$
secretly agrees on a key $K_{ij}$. This can be represented by a key
graph like the one shown in Figure~\ref{Flo:2}(a). By definition
$K_{ji}=-K_{ij}$ and $K_{ii}=0$. 
\item During the second step, each participant $P_{i}$ computes a ciphertext
$O_{i}$ by computing the sum of his secret keys; i.e., $O_{i}=\sum_{j}K_{ij}$.
The anonymous sender additionally adds his message $M$. This is illustrated
in Figure \ref{Flo:2}(b).
\end{enumerate}
Since $K_{ij}+K_{ji}=0$, all secret keys cancel in the sum $C$,
and only the message $M$ remains. When several participants try to
send a message during the same round, the messages collide (e.g. $C=M+M'+M''$)
and no meaningful data is transmitted.

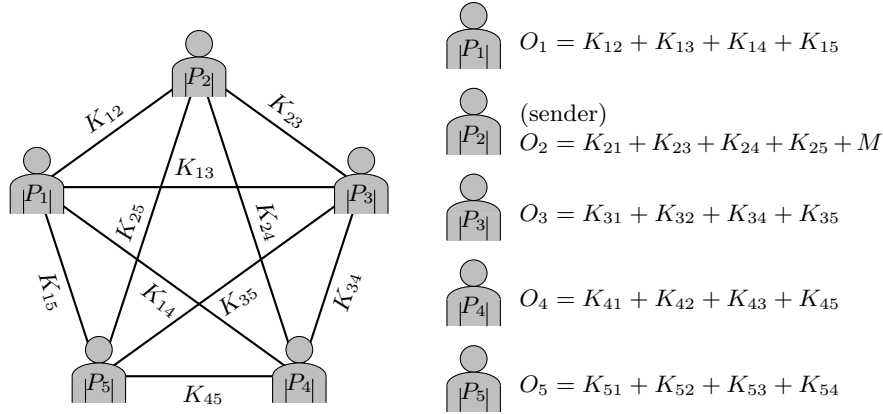
\begin{figure}[t]
\centering
\subfigure[Key graph showing which participants $P_i$ and $P_j$ share a secret key $K_{ij}$ (and its inverse  $K_{ji}:=-K_{ij}$).] { 
\begin{footnotesize}
\newcommand{\participK}[2]{
\begin{scope}[shift={(#1)}]
\draw [draw,fill=lightgray] (0pt,10pt) circle (5pt);
\draw [draw,fill=lightgray] (10pt,0pt) arc (0:180:10pt and 5pt);
\fill [lightgray] (-10pt,-10pt) rectangle (10pt,0pt);
\draw [draw] (10pt,0pt) -- (10pt,-10pt);
\draw [draw] (-10pt,0pt) -- (-10pt,-10pt);
\draw [draw] (5pt,-1pt) -- (5pt,-10pt);
\draw [draw] (-5pt,-1pt) -- (-5pt,-10pt);
\draw [anchor=center] (0pt,-2.5pt) node {#2};
\end{scope}
}
\begin{tikzpicture}
\begin{scope}
\path (18:6.9em) coordinate (P3);
\path (90:6.9em) coordinate (P2);
\path (162:6.9em) coordinate (P1);
\path (234:6.9em) coordinate (P5);
\path (306:6.9em) coordinate (P4);
\draw [thick] (P3) to node [anchor=south ,pos=0.5,swap,sloped] {$K_{23}\ $} (P2);
\draw [thick] (P3) to node [anchor=south ,pos=0.5,swap,sloped] {$K_{13}\ $} (P1);
\draw [thick] (P3) to node [anchor=north ,pos=0.5,sloped] {$K_{35}\ $} (P5);
\draw [thick] (P3) to node [anchor=north ,pos=0.5,sloped] {$K_{34}\ $} (P4);
\draw [thick] (P2) to node [anchor=south ,pos=0.5,swap,sloped] {$K_{12}\ $} (P1);
\draw [thick] (P2) to node [anchor=south ,pos=0.5,sloped] {$K_{25}\ $} (P5);
\draw [thick] (P2) to node [anchor=south ,pos=0.5,swap,sloped] {\ $K_{24}$} (P4);
\draw [thick] (P1) to node [anchor=north ,pos=0.5,sloped] {\ $K_{15}$} (P5);
\draw [thick] (P1) to node [anchor=north ,pos=0.5,swap,sloped] {\ $K_{14}$} (P4);
\draw [thick] (P5) to node [anchor=north ,pos=0.5,swap,sloped] {\ $K_{45}$} (P4);
\participK{P1}{$P_1$};
\participK{P2}{$P_2$};
\participK{P3}{$P_3$};
\participK{P4}{$P_4$};
\participK{P5}{$P_5$};
\end{scope}

\end{tikzpicture}
\end{footnotesize}
\label{2a} } \  
\subfigure[Ciphertexts $O_i$ are computed using the secret keys $K_{ij}$. The sender also add his message $M$.] { 
\begin{footnotesize}
\newcommand{\participK}[2]{
\begin{scope}[shift={(#1)}]
\draw [draw,fill=lightgray] (0pt,10pt) circle (5pt);
\draw [draw,fill=lightgray] (10pt,0pt) arc (0:180:10pt and 5pt);
\fill [lightgray] (-10pt,-10pt) rectangle (10pt,0pt);
\draw [draw] (10pt,0pt) -- (10pt,-10pt);
\draw [draw] (-10pt,0pt) -- (-10pt,-10pt);
\draw [draw] (5pt,-1pt) -- (5pt,-10pt);
\draw [draw] (-5pt,-1pt) -- (-5pt,-10pt);
\draw [anchor=center] (0pt,-2.5pt) node {#2};
\end{scope}
}
\begin{tikzpicture}

\begin{scope}
\path (0em,7em) coordinate (Q1);
\path (0em,3.5em) coordinate (Q2);
\path (0em,0em) coordinate (Q3);
\path (0em,-3.5em) coordinate (Q4);
\path (0em,-7em) coordinate (Q5);
\node [rectangle, right of=Q1, anchor=west, node distance=1.5em] {$O_1=K_{12}+ K_{13}+ K_{14}+ K_{15}$};
\node [rectangle, right of=Q2, anchor=west, node distance=1.5em, text width=136pt] {(sender)\\$O_2=K_{21}+ K_{23}+ K_{24}+ K_{25}+ M$};


\node [rectangle, right of=Q3, anchor=west, node distance=1.5em] {$O_3=K_{31}+ K_{32}+ K_{34}+ K_{35}$};
\node [rectangle, right of=Q4, anchor=west, node distance=1.5em] {$O_4=K_{41}+ K_{42}+ K_{43}+ K_{45}$};
\node [rectangle, right of=Q5, anchor=west, node distance=1.5em] {$O_5=K_{51}+ K_{52}+ K_{53}+ K_{54}$};
\participK{Q1}{$P_1$};
\participK{Q2}{$P_2$};
\participK{Q3}{$P_3$};
\participK{Q4}{$P_4$};
\participK{Q5}{$P_5$};
\end{scope}

\end{tikzpicture}
\end{footnotesize}
\label{2b} }

\protect\caption{Computation of ciphertexts in the dining cryptographers protocol. }

\label{Flo:2}
\end{figure}

A major problem of the protocol is that no communication can take
place if a malicious participant deliberately creates collisions all
the time. As the anonymity of the honest participants must not be
compromised, the detection of such a disruptor is difficult. While
computationally secure variation have been proposed, no efficient
and practically usable solution has been proposed for the information-theoretical
setting until today.

\subsubsection*{Related Work}

In \cite{golle2004dcr}, Golle and Juels used the Diffie-Hellman key
exchange to construct ciphertexts with an algebraic structure that
can be used in zero-knowledge proofs. More recently, Franck showed
in \cite{Franck_DC_0924} how to use such ciphertexts to detect cheating
participants in the context of collision resolution algorithms. However,
this approach does not offer the unconditional anonymity of the initial
protocol.\vspace{-2bp}

\subsubsection*{Our Contribution}

In this paper, we present a novel unconditionally untraceable dining
cryptographers scheme with optimal collision resolution, in which
computationally restricted disruptors are easy to detect. We use Pedersen
commitments to computationally bind participants to their secret keys
and then we use these commitments to prove the correct retransmission
of messages in a tree based collision resolution algorithm (We use
verifiable superposed receiving, presented by Pfitzmann in \cite{pfitzmann1990dkt}
and Waidner in \cite{waidner1990usa}). 

We believe our scheme is a significant improvement over the reservation
based technique of the initial dining cryptographers protocol \cite{chaum1988dcp},
wherein the detection of disruptors is lengthy and cumbersome. We
see possible applications in various areas like low latency anonymous
communication and electronic voting.\vspace{-2bp}

\subsubsection*{Outline of the Paper}

The rest of this paper is organized as follows. Section~2 contains
the preliminaries. In Section~3, we show how we extend the dining
cryptographers scheme with Pedersen commitments. In Section~4, we
show how the commitments can be used to construct statements that
can be used in zero-knowledge proofs. In Section~5, we show how to
implement verifiable collision resolution. Section~6 contains some
practical considerations. Section~7 is about related work, and Section~8
concluding remarks.

\pagebreak{}

\section{Preliminaries}

In this section, we discuss the assumptions and the primitives that
we use in the the rest of the paper.

\subsubsection*{Computational Assumptions}

We assume a group of $n$ participants $P_{1},...,P_{n}$ that can
be modeled by poly-time turing machines. We need a short-time computational
assumption to verify the correct execution of the protocol in zero-knowledge.
The anonymity of the transmitted data is unconditionally secure.

\subsubsection*{Communication Channels}

We assume a secure communication channel between each distinct pair
of participants $P_{i}$ and $P_{j}$, and we assume a fully connected
key graph.

We further assume a reliable synchronous broadcast channel \cite{lamport1982bgp},
which allows each participant $P_{i}\in\mathbb{P}$ to send a message
to all other participants. The recipients thus have the guarantee
that all then receive the same message, and that this is the same
unfalsified message that was send out by the sender.

\subsubsection*{Pedersen Commitments \cite{conf/crypto/Pedersen91}}

Let $G$ be a group of order $q$ in which the discrete logarithm
problem is assumed to be hard, and let $g$ and $h$ be randomly chosen
generators of a G. To commit to secret $K\in\mathbb{Z}_{q}$, the
committer choses random $r\in\mathbb{Z}_{q}$ and computes the commitment
\[
c:=g^{K}h^{r}.
\]
The committer can open the commitment by revealing $(K,r)$. Such
a commitment is \emph{unconditionally hiding}, which means that $K$
is perfectly secret until the commitment is opened. Further, such
a commitment is \emph{computationally binding}, which means that it
is computationally hard to find $(K',r')\ne(K,r)$, such that $c=g^{K'}h^{r'}$.
And finally, such commitments are \emph{homomorphic}; which means
that for $c=g^{K}h^{r}$ and $c'=g^{K'}h^{r'}$, we also have $c''=cc'=g^{K+K'}h^{r+r'}$.

\subsubsection*{Zero-Knowledge Proofs}

A zero-knowledge proofs allows a prover to convince a verifier that
he knows a witness which verifies a given statement, without revealing
the witness or giving the verifier any other information. One can
for instance construct a zero-knowledge proof to show the knowledge
of a discrete logarithm, the equality of discrete logarithms with
different bases, and logical $\wedge$ (and) and $\vee$ (or) combinations
thereof. A system for proving general statements about discrete logarithms
was presented in \cite{camenisch1997psg}. In our notation based on
\cite{camenisch1997egs}, secrets are represented by greek symbols. 
\begin{example}
A proof of knowledge of the discrete logarithm of $y$ to the base
$g$ as 
\[
\mathcal{PK}\{\alpha:y=g^{\alpha}\}.
\]
\pagebreak{}
\end{example}

\section{Extended Scheme with Pedersen Commitments}

In this section, we propose a way to extend the dining cryptographers
scheme using Pedersen commitments. We let each participant $P_{i}$,
$i\in\{1,...,n\}$ broadcast a tuple $(O_{i},c_{i})\in\mathbb{Z}_{q}\times G$
instead of just broadcasting $O_{i}$. The element $c_{i}$ is a Pedersen
commitment to the value $K_{i}$. The algebraic (discrete log based)
structure of $c_{i}$ will later allow to prove statements about $O_{i}$
in zero-knowledge. As $c_{i}$ is unconditionally hiding, the security
of the original protocol is preserved.

\subsubsection*{Detailed Description}

During the setup phase, when participants $P_{i}$ and $P_{j}$, $i\ne j$
agree on a secret key $K_{ij}\in\mathbb{Z}_{q}$, we require them
to additionally agree on a second secret value $r_{ij}\in\mathbb{Z}_{q}$.
Similarly to $K_{ji}=-K_{ij}$, we define $r_{ji}=-r_{ij}$. To simplify
the description we further define $r_{ii}=0$. The value $r_{ij}$
is then used by participant $P_{i}$ to commit to the secret key $K_{ij}$,
using the Pedersen commitment 
\[
c_{ij}:=g^{K_{ij}}h^{r_{ij}}.
\]
Note that $P_{i}$ and $P_{j}$ know the secrets $r_{ij}$ and $K_{ij}$,
so that both of them can compute and open $c_{ij}$. This knowledge
is used by $P_{j}$ to further provide $P_{i}$ with a digital signature
\[
\mathcal{S}_{j}(c_{ij}).
\]
This digital signature can later be used by $P_{i}$ to prove the
authenticity of $c_{ij}$ to a third party. Revealing $c_{ij}$ and
$\mathcal{S}_{j}(c_{ij})$ will not give away any information about
$K_{ij}$, since $c_{ij}$ is unconditionally hiding.

If participant $P_{i}$'s ciphertext $O_{i}$ does not contain a message,
we have
\[
O_{i}=K_{i}
\]
where
\[
K_{i}:=\sum_{j=1}^{n}K_{ij}.
\]
A Pedersen commitment for $K_{i}$ can be computed from $c_{i1},...,c_{in}$
according to
\[
c_{i}:=\prod_{j=1}^{n}c_{ij}.
\]
This aggregation of commitments is illustrated in Figure \ref{figcommit},
where $P_{1}$ computes the commitment $c_{1}$ for the ciphertext
$O_{1}=K_{1}$. This commitment $c_{1}$ could be opened by $P_{1}$
using $K_{1}$ and $\sum_{j=1}^{n}r_{1j}$.

\begin{figure}[t]

\begin{footnotesize}
\newcommand{\participK}[2]{
\begin{scope}[shift={(#1)}]
\draw [draw,fill=lightgray] (0pt,10pt) circle (5pt);
\draw [draw,fill=lightgray] (10pt,0pt) arc (0:180:10pt and 5pt);
\fill [lightgray] (-10pt,-10pt) rectangle (10pt,0pt);
\draw [draw] (10pt,0pt) -- (10pt,-10pt);
\draw [draw] (-10pt,0pt) -- (-10pt,-10pt);
\draw [draw] (5pt,-1pt) -- (5pt,-10pt);
\draw [draw] (-5pt,-1pt) -- (-5pt,-10pt);
\draw [anchor=center] (0pt,-2.5pt) node {#2};
\end{scope}
}
\begin{tikzpicture}
\begin{scope}
\path (18:7em) coordinate (P3);
\path (90:7em) coordinate (P2);
\path (162:7em) coordinate (P1);
\path (234:7em) coordinate (P5);
\path (306:7em) coordinate (P4);
\draw [dotted] (P3) to node [dotted,anchor=south ,pos=0.5,swap,sloped] {} (P2);
\draw [thick] (P3) to node [anchor=south ,pos=0.5,swap,sloped] {$\overbrace{K_{13},r_{13}}^{\displaystyle c_{13}}$} (P1);
\draw [dotted] (P3) to node [dotted,anchor=north ,pos=0.5,sloped] {} (P5);
\draw [dotted] (P3) to node [dotted,anchor=north ,pos=0.5,sloped] {} (P4);
\draw [thick] (P2) to node [anchor=south ,pos=0.5,swap,sloped] {$\overbrace{K_{12},r_{12}}^{\displaystyle c_{12}}$} (P1);
\draw [dotted] (P2) to node [dotted,anchor=south ,pos=0.5,sloped] {} (P5);
\draw [dotted] (P2) to node [dotted,anchor=south ,pos=0.5,swap,sloped] {} (P4);
\draw [thick] (P1) to node [anchor=north ,pos=0.5,sloped,swap] {$\underbrace{K_{15},r_{15}}_{\displaystyle c_{15}}$} (P5);
\draw [thick] (P1) to node [anchor=south ,pos=0.5,sloped] {$\overbrace{K_{14},r_{14}}^{\displaystyle c_{14}}$} (P4);
\draw [dotted] (P5) to node [dotted,anchor=north ,pos=0.5,swap,sloped] {} (P4);
\participK{P1}{$P_1$};
\participK{P2}{$P_2$};
\participK{P3}{$P_3$};
\participK{P4}{$P_4$};
\participK{P5}{$P_5$};
\end{scope}

\begin{scope}[xshift=105pt]
\path (0em,7em) coordinate (Q1);
\path (0em,3.5em) coordinate (Q2);
\path (0em,0em) coordinate (Q3);
\path (0em,-3.5em) coordinate (Q4);
\path (0em,-7em) coordinate (Q5);
\node [rectangle, right of=Q1, anchor=west, node distance=1.5em] {$O_1=\underbrace{K_{12}+ K_{13}+ K_{14}+ K_{15}}_{\displaystyle c_1=c_{12}\cdot c_{13}\cdot c_{14}\cdot c_{15}}$};


\participK{Q1}{$P_1$};
\end{scope}

\end{tikzpicture}
\end{footnotesize}

\protect\caption{Example: Binding to the secret keys using shared Pedersen commitments.
Participant $P_{1}$ agrees on the secret keys $K_{12},K_{13},K_{14},K_{15}\in\mathbb{Z}_{q}$
and the secret values $r_{12},r_{13},r_{14},r_{15}\in\mathbb{Z}_{q}$
respectively with the participants $P_{2},P_{3},P_{4}$ and $P_{5}$.
This allows him to compute the commitments $c_{12},c_{13},c_{14},c_{15}\in G$
with $c_{ij}=g^{K_{ij}}h^{r_{ij}}$. Finally, he computes $c_{1}=\prod_{j=2}^{5}c_{1j}$,
the commitment for $K_{1}:=\sum_{j=2}^{5}K_{1j}$.}
\label{figcommit}
\end{figure}
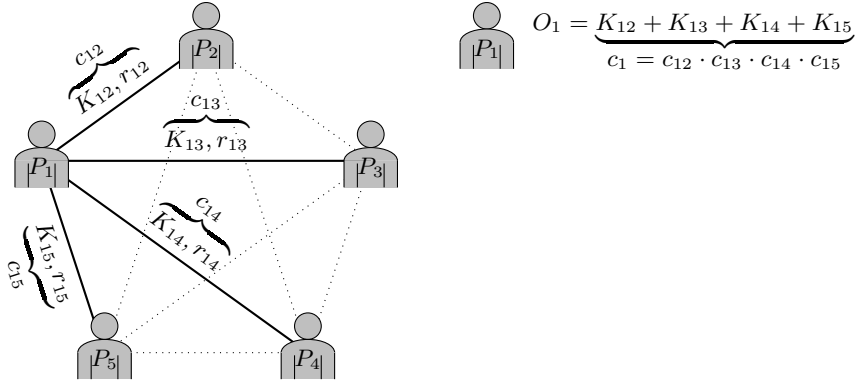

\medskip{}

During the broadcast phase, the participants $P_{1},...,P_{n}$ respectively
send the tuples $(O_{1},c_{1}),...,(O_{n},c_{n})$. The commitments
$c_{1},...,c_{n}$ are valid if 
\begin{equation}
\prod_{i=1}^{n}c_{i}=1.\label{eq:prop1}
\end{equation}
If (\ref{eq:prop1}) does not hold, this means that at least one participant
cheated. To find the cheater(s) an investigation phase can be performed.\medskip{}

During such an investigation phase, each participant $P_{i}$ will
publish the secret keys $c_{ij}$ and the corresponding signatures
$\mathcal{S}_{j}(c_{ij})$ for $j\in\{1,...,n\}\backslash\{i\}$.
The signatures have to be correct, and it must hold that
\begin{equation}
c_{i}=\prod_{j=1}^{n}c_{ij}\label{eq:a}
\end{equation}
and
\begin{equation}
c_{ij}c_{ji}=1.\label{eq:b}
\end{equation}
If a signature is wrong or if (\ref{eq:a}) or (\ref{eq:b}) does
not hold, then the corresponding participant $P_{i}$ cheated. The
fact that (\ref{eq:prop1}) and (\ref{eq:b}) must hold is because
we have $K_{ij}=-K_{ji}$ and $r_{ij}=-r_{ji}$ by construction, and
(\ref{eq:a}) must hold by definition.

\section{Statements for Zero-Knowledge Proofs}

In this section, we propose statements that can be used in zero-knowledge
proofs.

\subsubsection*{Statements about Single Rounds}

During a single round of the dining cryptographers protocol, a participant
$P_{i}$ broadcasts a ciphertext $(O_{i},c_{i})$. Either we have
$O_{i}=K_{i}$ or $O_{i}=K_{i}+M$. A statement that holds when $O_{i}$
does not encode a message $M$ is given in Theorem~\ref{thm:1}.
The proof is given in Appendix~\ref{appenix:Proofs}.
\begin{theorem}
\label{thm:1}If a poly-time participant $P_{i}$ generates the tuple
$(O_{i},c_{i})$ and $P_{i}$ knows $\alpha$ such that $c_{i}=g^{O_{i}}h^{\alpha}$,
then we have $O_{i}=K_{i}$.\pagebreak{}
\end{theorem}

\subsubsection*{Statements about Multiple Rounds}

In order to discuss multiple rounds, we use a superscript $^{(k)}$
to denote a value of a round~$k$. E.g., the values $O_{i}^{(1)}$,
$O_{i}^{(2)}$ and $O_{i}^{(3)}$ denote the ciphertexts broadcasted
by $P_{i}$ during the rounds 1, 2 and 3 respectively. 

Theorem~\ref{thm:If-a-poly-time} provides a statement that holds
when ciphertexts of two rounds encode the same message. 
\begin{theorem}
\label{thm:If-a-poly-time}If a poly-time participant $P_{i}$ generates
the tuples $(O_{i}^{(1)},c_{i}^{(1)})$ and $(O_{i}^{(2)},c_{i}^{(2)})$,
and $P_{i}$ knows $\alpha$ such that $c_{i}^{(1)}(c_{i}^{(2)})^{-1}=g^{O_{i}^{(1)}-O_{i}^{(2)}}h^{\alpha}$,
then $O_{i}^{(1)}$ and $O_{i}^{(2)}$ encode the same message. 
\end{theorem}

Theorem~\ref{thm:If-a-poly-time-1} provides a statement that holds
when a message encoded in a first ciphertext is encoded at most once
in a series of other ciphertexts (while the rest of the ciphertexts
does not encode a message).
\begin{theorem}
\label{thm:If-a-poly-time-1}If a poly-time participant $P_{i}$ generates
 $(O_{i}^{(1)},c_{i}^{(1)}),...,(O_{i}^{(l)},c_{i}^{(l)})$, and $P_{i}$
knows $\alpha$ such that
\begin{equation}
\bigwedge_{k=2}^{l}\left(\left(c_{i}^{(1)}\prod_{j=2}^{k}(c_{i}^{(j)})^{-1}=g^{O_{i}^{(1)}-\sum_{j=2}^{k}O_{i}^{(j)}}h^{\alpha}\right)\vee\left(c=g^{O_{i}^{(k)}}h^{\alpha}\right)\right)\label{eq:theorem}
\end{equation}
then at most one ciphertext of $O_{i}^{(2)},...,O_{i}^{(l)}$ encodes
the same message as $O_{i}^{(1)}$, while the other ciphertexts of
$O_{i}^{(2)},...,O_{i}^{(l)}$ encode no message.
\end{theorem}

The statements from the preceding theorems can be used in zero-knowledge
proofs. We will see in the next section how we can use this for proving
the correct execution of a collision resolution algorithm.

\section{Implementing Verifiable Superposed Receiving }

Superposed receiving is a collision resolution scheme for the Dining
Cryptographers protocol proposed by Pfitzmann in \cite{pfitzmann1990dkt}
and Waidner in \cite{waidner1990usa}. It achieves an optimal throughput
of one message per round. However, the scheme was never used in practice,
as a malicious participant may disrupt the process and remain undetected.
In this section, we show that in our scheme such disruptors are easy
to detect.
\begin{figure}[t]
\begin{centering}
\begin{tikzpicture}[level distance=1.0cm,
sibling distance=5cm,level/.style={sibling distance=1.2cm/#1},
block/.style ={rectangle, draw=black, thick, text centered,  inner sep=0.15cm,text height=1.8ex,minimum width=5ex,font=\small},
blockx/.style ={rectangle, fill=lightgray, draw=black, thick, text centered,  inner sep=0.15cm,font=\small}]

\begin{scope}
\node[block] (a1){}
child {node (a2)[block] {}}
child {node (a3)[block,dashed,yshift=-1cm] {}};
\end{scope}

\begin{scope}[xshift=3.5cm]
\node[block] (b1){$M$}
child {node (b2)[block] {$M$}}
child {node (b3)[block,dashed,yshift=-1cm] {}};
\end{scope}

\begin{scope}[xshift=7cm]
\node[block] (c1){$M$}
child {node (c2)[block] {}}
child {node (c3)[block,dashed,yshift=-1cm] {$M$}};
\end{scope}

\node[inner sep=0cm,fit=(a1) (a2) (a3)] (la) {};
\node at (la.south)[below=0.3cm,inner sep=0,font=\small] {(a) No message.};
\node[inner sep=0cm,fit=(b1) (b2) (b3)] (lb) {};
\node at (lb.south)[below=0.3cm,inner sep=0,font=\small, text width=3cm] {(b) Retransmit in round $2k$.};
\node[inner sep=0cm,fit=(c1) (c2) (c3)] (lc) {};
\node at (lc.south)[below=0.3cm,inner sep=0,font=\small, text width=3cm] {(c) 'Retransmit' in round $2k+1$.};

\draw [densely dotted]($(a1)+ (-3.0cm,0.5cm)$) node [above,xshift=0.7cm]{round id}-- +(12.1cm,0);
\draw [densely dotted]($(a1)+ (-3.0cm,-0.5cm)$) node [above,xshift=0.7cm]{$k$}-- +(12.1cm,0);
\draw [densely dotted]($(a1)+ (-3.0cm,-1.5cm)$) node [above,xshift=0.7cm]{$2k$}-- +(12.1cm,0);
\draw [densely dotted]($(a1)+ (-3.0cm,-2.5cm)$) node [above,xshift=0.7cm]{$2k+1$}-- +(12.1cm,0);

\end{tikzpicture}

\par\end{centering}

\protect\caption{Retransmission in superposed receiving. Only message involved in the
in the collision in round $k$ may be retransmitted either in round
$2k$. No new message may be sent during the collision resolution
process.}
\label{figure1}
\end{figure}
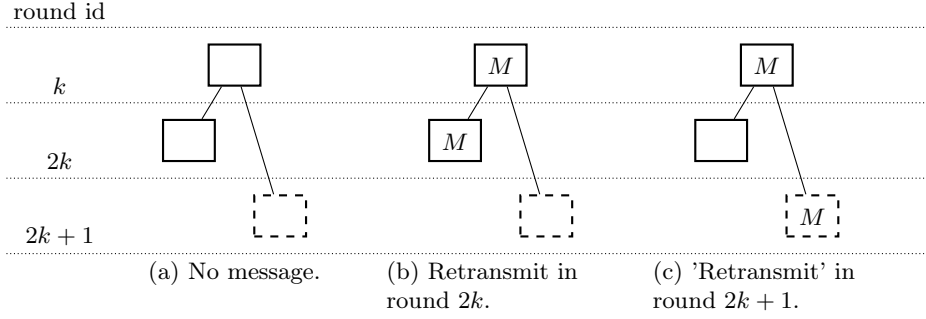

\subsubsection*{Superposed Receiving}

A collision occurs when multiple participants send a message in the
same round. In superposed receiving, collisions are repeatedly split
in two, until all messages are transmitted. An exemplary collision
resolution tree is shown in Figure~\ref{supreceiving_tree-1}. To
keep our description simple, we assume that when a collision occurs
in a round $k$, the rounds $2k$ and $2k+1$ are used to split this
collision. Like in the previous section, we use subscripts to denote
values of the different rounds, e.g. $O^{(7)}$ for ciphertext of
round 7.

In superposed receiving messages are tuples of the form $(1,m)$.
It is then possible to determine the number of messages involved in
a collision and to compute the average value of a message involved
in the collision. For instance, when 3 messages $(1,m)$, $(1,m')$
and $(1,m'')$ collide in round $k$, the tuple $(3,m+m'+m'')$ is
received and the average value is then $(m+m'+m'')/3$. Then, in round
$2k$ only the messages with a value below this average are retransmitted,
while the rest of the messages goes to round $2k+1$. To keep our
description simple, we assume that a tuple of the form $(1,m)$ is
encoded in a message $M\in\mathbb{Z}_{q}$, such that the individual
elements of the tuple are added when there is a collision.

It is not necessary to transmit anything in round $2k+1$. Instead,
the result of round $2k+1$ is inferred by subtracting the results
of round $2k$ from the result of round $k$. This technique, which
is also known as inference cancellation \cite{yu2005sicta}, is the
reason for the optimal throughput of the scheme. In the example of
Figure~\ref{supreceiving_tree-1} only 5 rounds are transmitted for
5 messages. For inference cancellation to work, the collision resolution
must operate in blocked access mode. This means that no new message
may be sent by any participant until the collision resolution process
is over.

\subsubsection*{Verification}

A malicious participant may try to disrupt the collision resolution
process by not properly participating in the collision resolution
process. We verify the correct execution of the protocol in two steps.

First, we verify in round $2k$ that, according to Figure~\ref{figure1},
a participant either retransmits exactly the same message that he
sent in round $k$, or that he sends no message at all. Using the
statements from the previous section, each participant can prove that
his ciphertext $O_{2k}$ is correct, without revealing if whether
it contains a message or not. To do this, the participant generates
a zero-knowledge proof that proves that he knows $\alpha$ such that
\begin{equation}
\mathcal{PK}\{\alpha:(c_{i}^{(2k)}=g^{O_{i}^{(2k)}}h^{\alpha})\vee(c_{i}^{(k)}c_{i}^{(2k)-1}=g^{O_{i}^{(k)}-O_{i}^{(2k)}}h^{\alpha})\}\label{eq:33}
\end{equation}
holds. With this proof he can convince a verifier that he participated
correctly, without compromising the anonymity of the protocol. As
described before, in some rounds no transmission takes place and so
there might no be a value $O_{i}^{(k)}$ available to prove the correctness
of $O_{i}^{(2k)}$ using  statement (\ref{eq:33}). However, it is
still possible to prove that $O_{i}^{(2k)}$ is correct by proving
that a message contained in the nearest transmitted parent round is
transmitted at most once in all the branches down to $O_{i}^{(2k)}$.
This can be done using Theorem~\ref{thm:If-a-poly-time-1}. 

\begin{figure}
\begin{centering}
\footnotesize
\pgfdeclarelayer{background layer}
\pgfdeclarelayer{foreground layer}
\pgfsetlayers{background layer,main,foreground layer}
\begin{tikzpicture}[yscale=0.8]
\begin{scope}[level distance=2.0cm,
sibling distance=4cm,level/.style={sibling distance=3.8cm/#1},
edge from parent/.style={draw,very thick},
block/.style ={rectangle, draw=black, thick, text centered,  inner sep=0.12cm,font=\footnotesize},
blockx/.style ={rectangle, dashed, draw=black, very thick, text centered,  inner sep=0.12cm,font=\footnotesize}]
\begin{pgfonlayer}{foreground layer}
\node[block] (n1){$\underbrace{M_1+M_2+M_3+M_4+M_5}_{\displaystyle (5,130)}$}
child {node (n21)[block] {$\underbrace{M_2+M_4}_{\displaystyle (2,28)}$} 
child {node (n2x)[block,yshift=-1.6cm] {$\underbrace{M_2}_{\displaystyle (1,11)}$} edge from parent node[fill=white,anchor=east,xshift=-0.05cm] {$<14 $}}
child {node (n2a)[blockx,yshift=-3.2cm] {$\underbrace{M_4}_{\displaystyle (1,17)}$}}
edge from parent node[fill=white,anchor=east,xshift=-0.25cm] {$<26$}}
child {node (n22)[blockx,yshift=-1.6cm] {$\underbrace{M_1+M_3+M_5}_{\displaystyle (3,102)}$}
child {node (n31)[block,yshift=-3.2cm] {$\underbrace{M_3}_{\displaystyle (1,28)}$}edge from parent node[fill=white,anchor=east,xshift=-0.05cm] {$<34 $~}}
child {node (n32)[blockx,yshift=-4.8cm] {$\underbrace{M_1+M_5}_{\displaystyle (2,74)}$}
child {node (n3f)[block] {$\underbrace{M_1}_{\displaystyle (1,36)}$}edge from parent node[fill=white,anchor=east,xshift=-0.1cm] {$<37 $~}}
child {node (n3g)[blockx,yshift=-1.6cm] {$\underbrace{M_5}_{\displaystyle (1,38)}$}}}
};
\end{pgfonlayer}



\path []($(n1)+ (-5.0cm,1.0cm)$) -- node[above]{$C^{(k)}$} +(10cm,0);

\begin{pgfonlayer}{background layer}
\draw [densely dotted]($(n1)+ (-5.0cm,1cm)$) node [above,xshift=0.8cm]{round id $k$}-- +(12.1cm,0);
\draw [densely dotted]($(n1)+ (-5.0cm,-1cm)$) node [above,xshift=0.8cm]{$1$}-- +(12.1cm,0) node[above,anchor=south west,xshift=-2.5cm]{$C^{(1)}=\sum^n_{i=1}O^{(k)}_1$};
\draw [densely dotted]($(n1)+ (-5.0cm,-3cm)$) node [above,xshift=0.8cm]{$2$}-- +(12.1cm,0) node[above,anchor=south west,xshift=-2.5cm]{$C^{(2)}=\sum^n_{i=1}O^{(k)}_2$};
\draw [densely dotted]($(n1)+ (-5.0cm,-5cm)$) node [above,xshift=0.8cm]{$3$}-- +(12.1cm,0) node[above,anchor=south west,xshift=-2.5cm]{$C^{(3)}=C^{(1)}-C^{(2)}$};
\draw [densely dotted]($(n1)+ (-5.0cm,-7cm)$) node [above,xshift=0.8cm]{$4$}-- +(12.1cm,0) node[above,anchor=south west,xshift=-2.5cm]{$C^{(4)}=\sum^n_{i=1}O^{(k)}_4$};
\draw [densely dotted]($(n1)+ (-5.0cm,-9cm)$) node [above,xshift=0.8cm]{$5$}-- +(12.1cm,0) node[above,anchor=south west,xshift=-2.5cm]{$C^{(5)}=C^{(2)}-C^{(4)}$};
\draw [densely dotted]($(n1)+ (-5.0cm,-11cm)$) node [above,xshift=0.8cm]{$6$}-- +(12.1cm,0) node[above,anchor=south west,xshift=-2.5cm]{$C^{(6)}=\sum^n_{i=1}O^{(k)}_6$};
\draw [densely dotted]($(n1)+ (-5.0cm,-13cm)$) node [above,xshift=0.8cm]{$7$}-- +(12.1cm,0) node[above,anchor=south west,xshift=-2.5cm]{$C^{(7)}=C^{(3)}-C^{(6)}$};
\draw [densely dotted]($(n1)+ (-5.0cm,-15cm)$) node [above,xshift=0.8cm]{$14$}-- +(12.1cm,0) node[above,anchor=south west,xshift=-2.5cm]{$C^{(14)}=\sum^n_{i=1}O^{(k)}_{14}$};
\draw [densely dotted]($(n1)+ (-5.0cm,-17cm)$) node [above,xshift=0.8cm]{$15$}-- +(12.1cm,0) node[above,anchor=south west,xshift=-2.5cm]{$C^{(15)}=C^{(7)}-C^{(14)}$};
\end{pgfonlayer}

\node[inner sep=0cm,fit=(n1) (n3g) (n2x)] (la) {};
\end{scope}

\end{tikzpicture}
\par\end{centering}

\protect\caption{Exemplary binary collision resolution tree with superposed receiving.
In rounds 1,2,4,6 and 14, ciphertexts $O^{(k)}$ are transmitted,
and $C^{(k)}$ is computed using these ciphertexts. In rounds 3,5,7
and 15, no data is transmitted and $C^{(k)}$ is computed using data
from the parent and the sibling node.}
\label{supreceiving_tree-1}
\end{figure}
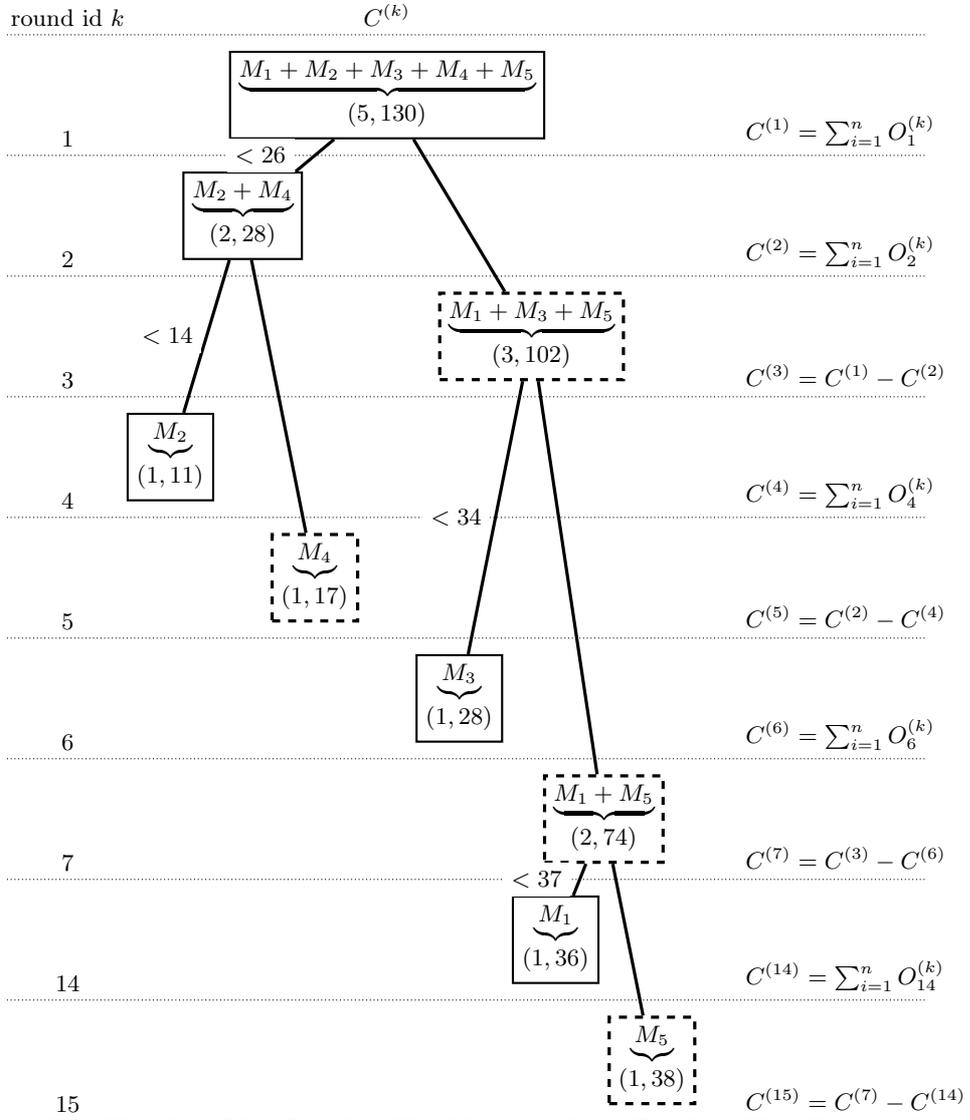

\begin{example}
In the collision resolution process shown in Figure \ref{supreceiving_tree-1},
a participant proves for $O^{(2)}$ that 
\[
\mathcal{PK}\{\alpha:(c_{i}^{(2)}=g^{O_{i}^{(2)}}h^{\alpha})\vee(c_{i}^{(1)}c_{i}^{(2)-1}=g^{O_{i}^{(1)}-O_{i}^{(2)}}h^{\alpha})\}
\]
holds, then for $O^{(4)}$ that
\[
\mathcal{PK}\{\alpha:(c_{i}^{(4)}=g^{O_{i}^{(4)}}h^{\alpha})\vee(c_{i}^{(2)}c_{i}^{(4)-1}=g^{O_{i}^{(2)}-O_{i}^{(4)}}h^{\alpha})\}
\]
holds, then for $O^{(6)}$ that
\[
\mathcal{PK}\{\alpha:(c_{i}^{(6)}=g^{O_{i}^{(6)}}h^{\alpha})\vee(c_{i}^{(1)}c_{i}^{(2)-1}c_{i}^{(6)-1}=g^{O_{i}^{(1)}-O_{i}^{(2)}-O_{i}^{(6)}}h^{\alpha})\}
\]
holds, then for $O^{(14)}$ that
\[
\mathcal{PK}\{\alpha:(c_{i}^{(14)}=g^{O_{i}^{(14)}}h^{\alpha})\vee(c_{i}^{(1)}c_{i}^{(2)-1}c_{i}^{(6)-1}c_{i}^{(14)-1}=g^{O_{i}^{(1)}-O_{i}^{(2)}-O_{i}^{(6)}-O_{i}^{(14)}}h^{\alpha})\}
\]
holds. 
\end{example}

This shows that it is possible to verify that a participant retransmitted
his message in only one branch of the tree.

Then, we verify that every properly collision splits into 2 parts.
As we know that every collision is supposed to split, we know if all
messages end up in the same branch then at least one participant cheated.
So this can only happen when a malicious node retransmits message
in the wrong branch, or when the message is not of the correct form
$(1,m)$ initially. If such activity is detected, it is possible to
identify the disruptor by falling fall back to probabilistic splitting
of collisions \cite{Franck_DC_0924,yu2005sicta}. Each participant
then choses randomly whether to retransmit his message in round $2j$
or round $2j+1$. This allows to separate the honest nodes from the
malicious ones after a few rounds. After this separation has taken
place it is possible to determine the messages that have not been
transmitted in the right branch earlier, and to identify the corresponding
participants using a zero-knowledge proof (I.e. each participant has
to prove in zero-knowledge that he did not send the message that appeared
in the wrong branch.). If a collision repeatedly does not split, even
with probabilistic retransmission, then the involved participants
can be considered to be malicious. 

A disruptor will thus always be detected and can be banned from the
group of participants.

\subsubsection{}

\section{Practical Considerations}

This section contains a few remarks about various aspects of a practical
implementation.

\paragraph*{(Signatures with Merkle Trees)}

During the setup, participants mutually authenticate their commitments
$c_{ij}$ and $c_{ji}$ using the signatures $\mathcal{S}_{j}(c_{ij})$
and $\mathcal{S}_{i}(c_{ji})$. For many rounds, this can be implemented
more efficiently using a Merkle tree.

\paragraph*{(Mutual Signatures Attack)}

During the setup phase one participant could refuse to agree an a
shared secret with another participant. I.e., one participant could
refuse to provide a signature $\mathcal{S}(.)$ to the other participant.
As a response to this, we suggest that each participant may just publicly
claim that he is not sharing a secret key with the other participant.
It is then assumed the corresponding $K_{ij}=0$ and $c_{ij}=1$,
and no signature is required.

\paragraph*{(Efficient Investigation in Packet-Switched Networks)}

In order to efficiently detect disruptors, a single (trusted) investigator
can collect and verify the proofs. Only when he detects a cheater,
he will provide all other participants with a copy of the relevant
data.

\paragraph*{(Long Messages)}

To keep the description simple, we assumed that messages fit in a
single element of $\mathbb{Z}_{q}$; i.e., that a single $K$ is sufficient
for one round. For longer messages one can, as shown in \cite{PedersenExtended},
use a randomly chosen generator tuple $g,g',g'',...,h$ to commit
to a vector $(K,K',K'',...)$ by computing 
\[
c=g^{K}g'^{K'}g'^{K''}...h^{r}.
\]

\paragraph*{(Key Establishment)}

To obtain information-theoretical security from the protocol, it is
necessary to use real random secret keys for each round. In practice,
it is also possible to realize weaker system, where the shared secrets
are generated for instance using the Diffie-Hellman protocol.

\section{Concluding Remarks}

We have shown how to extend the dining cryptographers scheme with
Pedersen commitments, such that it is possible to construct zero-knowledge
proofs about the retransmission of data, without compromising the
anonimity of the protocol.

It is remarkable that it is then possible realize a verifiable dining
cryptographers protocol with an optimal throughput, which does not
require any kind of reservation phase prior to the transmission of
the messages. 

We believe that our approach is a significant step forward towards
the efficient implementation of unconditionally untraceable communication
systems.

We see possible applications in many fields, like low-latency untraceable
communication and secret shuffling. The main problem that remains
in practice is the secure agreement on secret keys between participants.

\bibliographystyle{plain}
\bibliography{references}

\appendix

\section{Proofs\label{appenix:Proofs}}
\begin{lemma}
\label{lem:Pedersen1-1}Given randomly chosen generators $g,h$ of
a group in which the discrete log problem is assumed to hold, a poly-time
participant can only find $(a,b),(a',b')$, such that when $g^{a}h^{b}=g^{a'}h^{b'}$,
it must hold that $a=a'$.\end{lemma}

\begin{proof}
If a poly-time participant can find $(a,b),(a',b')$ such that $g^{a}h^{b}=g^{a'}h^{b'}$
with $a\ne a'$, then he can also compute the discrete logarithm $\log_{h}g$
with 
\[
\log_{h}g=(b'-b)/(a-a').
\]
As this is impossible by assumption, the statement follows.\hfill{}$\square$\end{proof}

\begin{theorem}
If a poly-time participant $P_{i}$ generates the tuple $(O_{i},c_{i})$
and $P_{i}$ knows $\alpha$ such that $c_{i}=g^{O_{i}}h^{\alpha}$,
then we have $O_{i}=K_{i}$.\end{theorem}

\begin{proof}
By definition, we have $c_{i}=g^{K_{i}}h^{r_{i}}$. If poly-time participant
$P_{i}$ knows $O_{i}$ and $\alpha$, such that $c_{i}=g^{O_{i}}h^{\alpha}$,
it follows from Lemma \ref{lem:Pedersen1-1} that $O_{i}=K_{i}$.\hfill{}$\square$
\begin{theorem}
If a poly-time participant $P_{i}$ generates the tuples $(O_{i}^{(1)},c_{i}^{(1)})$
and $(O_{i}^{(2)},c_{i}^{(2)})$, and $P_{i}$ knows $\alpha$ such
that $c_{i}^{(1)}(c_{i}^{(2)})^{-1}=g^{O_{i}^{(1)}-O_{i}^{(2)}}h^{\alpha}$,
then $O_{i}^{(1)}$ and $O_{i}^{(2)}$ encode the same message. \end{theorem}

\begin{proof}
By definition, we have $O_{i}^{(2)}=K_{i}^{(2)}+M_{a}$ and $O_{i}^{(1)}=K_{i}^{(1)}+M_{b}$,
where $M_{a},M_{b}\in\mathbb{Z}_{q}$. Further, we have 
\begin{eqnarray*}
c_{i}^{(1)}(c_{i}^{(2)})^{-1} & = & g^{O_{i}^{(1)}-O_{i}^{(2)}}h^{\alpha}\\
g^{K_{i}^{(1)}-K_{i}^{(2)}}h^{r_{i}^{(1)}-r_{i}^{(2)}} & = & g^{O_{i}^{(1)}-O_{i}^{(2)}}h^{\alpha}.
\end{eqnarray*}
 According to Lemma \ref{lem:Pedersen1-1} it follows that $K_{i}^{(1)}-K_{i}^{(2)}=O_{i}^{(1)}-O_{i}^{(2)}$
and thus 
\[
M_{a}=M_{b},
\]
which is the statement.\hfill{}$\square$
\end{proof}

\end{proof}

\begin{theorem}
If a poly-time participant $P_{i}$ generates  $(O_{i}^{(1)},c_{i}^{(1)}),...,(O_{i}^{(l)},c_{i}^{(l)})$,
and $P_{i}$ knows $\alpha$ such that
\begin{equation}
\bigwedge_{k=2}^{l}\left(\left(c_{i}^{(1)}\prod_{j=2}^{k}(c_{i}^{(j)})^{-1}=g^{O_{i}^{(1)}-\sum_{j=2}^{k}O_{i}^{(j)}}h^{\alpha}\right)\vee\left(c=g^{O_{i}^{(k)}}h^{\alpha}\right)\right)\label{eq:theorem-1}
\end{equation}
then at most one ciphertext of $O_{i}^{(2)},...,O_{i}^{(l)}$ encodes
the same message as $O_{i}^{(1)}$, while the other ciphertexts of
$O_{i}^{(2)},...,O_{i}^{(l)}$ encode no message.\end{theorem}

\begin{proof}
With $c_{i}:=g^{K_{i}}h^{r_{i}}$ and Lemma \ref{lem:Pedersen1-1}
it follows that when (\ref{eq:theorem-1}) holds, we have
\[
\bigwedge_{k=2}^{l}\left(\left(O_{i}^{(1)}-K_{i}^{(1)}=\sum_{j=2}^{k}O_{i}^{(j)}-K_{i}^{(j)}\right)\vee(O_{i}^{(k)}=K_{i}^{(k)})\right).
\]

Assume the cipertext $O_{1}$ encodes the message $M$, so that $O_{1}=K_{1}+M$
(with possibly $M=0$). For $k=2$, we can then have either $O_{k}=K_{k}+M$
or $O_{k}=K_{k}$. For $k>2$ and $\sum_{j=2}^{k-1}O_{j}-K_{j}=0$,
we can have either $O_{k}=K_{k}+M$ or $O_{k}=K_{k}$. For $k>2$
and $\sum_{j=2}^{k-1}O_{j}-K_{j}=M$, we must have $O_{k}=K_{k}$.
That is, for increasing $k$, as long as $O_{2},...,O_{k-1}$ contains
no message, we can have either $O_{k}=K_{k}+M$ or $O_{k}=K_{k}$.
Once one ciphertext of $O_{2},...,O_{k-1}$ contains the message $M$,
we must have $O_{k}=K_{k}$. Thus, at most one ciphertext of $O_{2},...,O_{l}$
may encode the message $M$ encoded in $O_{1}$, while the other ones
contain no message, which is the statement.\hfill{}$\square$\end{proof}

\end{document}